\documentclass[aps,prl,twocolumn,amsmath,amssymb,showpacs,superscriptaddress,notitlepage,longbibliography]{revtex4-1}
\usepackage[colorlinks=true,linkcolor=blue,anchorcolor=red,citecolor=blue, urlcolor=blue]{hyperref}
\usepackage{bm}
\usepackage{graphicx}
\usepackage{color}
\usepackage{xcolor}
\usepackage{soul}
\usepackage{mathrsfs}
\usepackage{braket}

\usepackage{amsthm}
\newtheorem{thm}{Theorem}
\newtheorem{dfn}[]{Definition}

\newtheorem{lemma}[]{Lemma}

\newtheorem{proposition}[]{Proposition}

\begin{document}
\title{On the energy barrier of hypergraph product codes}

\author{Guangqi Zhao}
\affiliation{Centre for Engineered Quantum Systems, School of Physics,
University of Sydney, Sydney, NSW 2006, Australia}
\email{gzha4233@uni.sydney.edu.au}

\author{Andrew C. Doherty}
\affiliation{Centre for Engineered Quantum Systems, School of Physics,
University of Sydney, Sydney, NSW 2006, Australia}

\author{Isaac H. Kim}
\affiliation{Department of Computer Science, University of California, Davis, CA 95616, USA}

\begin{abstract}
A macroscopic energy barrier is a necessary condition for self-correcting quantum memory. 
In this paper, we prove tight bounds on the energy barrier applicable to any quantum code obtained from the hypergraph product of two classical codes. If the underlying classical codes are low-density parity-check codes (LDPC), the energy barrier of the quantum code is shown to be the minimum energy barrier of the underlying classical codes (and their transposes) up to an additive $O(1)$ constant.
\end{abstract}

\maketitle

{\color{blue}\emph{Introduction}.}- Quantum computers are capable of solving problems that are likely intractable for classical computers, such as factoring of large integers~\cite{Shor1994} and simulation of quantum systems~\cite{Feynman1982,Lloyd1996}. However, realistic quantum computers are noisy. In order to build a useful quantum computer capable of carrying out such computational tasks, quantum error correction is likely necessary.

Since the discovery of the first quantum error-correcting code~\cite{shor1996fault}, much progress has been made towards finding better codes. While the leading approach to scalable quantum error correction has been the surface code~\cite{kitaev2003fault} for nearly 20 years, lately there has been a surge of interest in using quantum low-density parity check (LDPC) codes. This recent interest is in part due to Gottesman, who showed that with quantum LDPC codes, one can achieve a constant overhead for fault-tolerant quantum computation~\cite{gottesman2014faulttolerant}. A well-known approach to construct such codes is the hypergraph product construction~\cite{2014-Tillich}. More recent studies led to the development of other families of quantum LDPC codes with improved parameters~\cite{freedman2002z2,evra2022decodable,kaufman2021new,hastings2021fiber,panteleev2021quantum,Breuckmann_2021,Panteleev.2021,leverrier2022quantum,dinur2022good}. Moreover, recent studies suggest that there are viable fault-tolerant quantum computing architectures that can host such codes  ~\cite{Tremblay2022,bravyi2024highthreshold,Xu2023}.

One challenge in using these codes lies in the decoding. While there are general methods such as the BP+OSD decoder~\cite{Panteleev2021degeneratequantum}, it is a priori not obvious how well such a general-purpose decoder would work for a given code. One attractive approach is to use codes that have an extensive energy barrier. For such codes, a simple process that iteratively lowers the energy (quantified in terms of the number of stabilizers violated) is a viable decoder candidate. Alternatively, such a model can be viewed as a candidate for self-correcting quantum memory, protecting quantum information by the macroscopic energy barrier. (We note that the energy barrier is not a sufficient condition for building a self-correcting quantum memory~\cite{Haah2011,Michnicki2014,Siva2017}, though it is nevertheless a necessary condition.) 

Such approaches work well for four-dimensional (4D) toric code~\cite{alicki2008} and the quantum expander code~\cite{leverrier2015quantum,fawzi2020constant}, both of which have extensive energy barriers. Unfortunately, rigorous bounds on the energy barrier are hard to come by and often derived for specific codes (or family of codes)~\cite{Bravyi.2011,Michnicki2014,leverrier2015quantum,Williamson2023,Lin2023}. 

In this paper, we prove a tight bound on the energy barrier for the hypergraph product codes~\cite{2014-Tillich}, defined in terms of the energy barrier of the underlying classical codes. The hypergraph product is defined in terms of the parity check matrices (or equivalently, the Tanner graphs) of two classical codes. While there are codes with better parameters~\cite{hastings2021fiber,Breuckmann_2021,panteleev2021quantum,Panteleev.2021}, the hypergraph product remains a flexible framework for constructing quantum LDPC codes with many advantages, such as the variety of decoders~\cite{leverrier2015quantum,fawzi2020constant,Panteleev2021degeneratequantum,Roffe2020,Grospellier2021}, logical gates~\cite{Krishna2021,Cohen2022,Quintavalle2023}, and distance-preserving syndrome extraction circuit~\cite{Tremblay2022,manes2023distance}.

Now we describe our main result. Without loss of generality, consider a hypergraph product of two classical codes, defined in terms of the parity check matrices $H_1$ and $H_2$. We denote the parity check matrix of the resulting quantum code as $H_{(H_1, H_2)}$. For both the quantum and the classical code, we denote the energy barrier as $\Delta(H)$, where $H$ can be a parity check matrix of the quantum or the classical code. We prove that under a modest condition,
\begin{equation}
    \Delta(H_{(H_1, H_2)}) = \min (\Delta(H_1), \Delta(H_2), \Delta(H_1^T), \Delta(H_2^T)), \label{eq:main_result}
\end{equation}
where $H^T$ is the transpose of $H$. Eq.~\eqref{eq:main_result} holds if the energy barriers of the classical codes are larger than or equal to a certain sparsity parameter of the code. Importantly, if the underlying codes are LDPC, then the sparsity parameter is $O(1)$. Therefore, Eq.~\eqref{eq:main_result} holds if the energy barriers of $H_1, H_2, H_1^T,$ and $H_2^T$ grow as the code size grows. If this condition is not satisfied, the energy barrier is bounded by a constant, a fact that follows trivially from the definition of the hypergraph product code.

The proof of Eq.~\eqref{eq:main_result} is based on two results. First, if two logical operators of a quantum LDPC code are equivalent up to a stabilizer, their energy barriers differ only by a constant related to the code's sparsity parameters $w_c, w_q$ [Theorem~\ref{thm:energy_barrier_bound_stabilizer}]. Thanks to this result, proving the energy barrier for a given code reduces to proving the energy barrier for any complete set of logical operators, which can be a much smaller set than the set of all logical operators. We then identify a set of logical operators for which the exact energy barrier can be determined in terms of the energy barriers of the underlying classical codes. Together, these two results imply Eq.~\eqref{eq:main_result}.

{\color{blue}\emph{Energy barrier of codes}.}- 
We first present a formal definition of the energy barrier for stabilizer codes~\cite{Bravyi2009nogo}. Let $\mathcal{C}$ be the code subspace of a stabilizer code, defined in terms of the stabilizer group $S$. The code subspace can be viewed as the ground state subspace of a Hamiltonian of the form $\hat{H} = \sum_{i=1}^{m}(I-s_i)/2$, where $\{s_1, \ldots, s_m\} \subset S$ is a set of generators. Note that the energy barrier depends on the choice of the generating set. For an operator $P$ in the Pauli group $\mathcal{P}$, the energy of the state $P|\psi\rangle$ is given by $\bra{\psi}P^\dag \hat{H} P\ket{\psi} = \epsilon(P)$. Here $\ket{\psi}$ is any ground state with $\bra{\psi}\hat{H}\ket{\psi} = 0$, and $\epsilon(P)$ is the energy cost of $P$. This is the number of $s_i$s that anticommute with $P$, i.e., $\epsilon(P) = | \{i: s_iP = -Ps_i\}|$. Equivalently, one may define the energy barrier in terms of the parity check matrix $H$ of the stabilizer code. Let $v(P)$ be the binary (bit-string) representation of a Pauli $P$:
\begin{equation}
    \epsilon(P) = \text{wt}(Hv(P)).
\end{equation}

A sequence $P_0, P_1, \ldots, P_t$ from the Pauli group $\mathcal{P}$ forms a \emph{path} from $P_0$ to $P_t$ if for each index $i$, the operators $P_i$ and $P_{i+1}$ differ at no more than one qubit. The notation $w(P_0, P_t)$ represents the collection of all such paths from $P_0$ to $P_t$. For $ r \in w(P_0, P_t)$, $\epsilon_{\max }(r)$ denotes the highest energy along path $r$, i.e., $\epsilon_{\max}(r) = \max_{P_i \in r}\epsilon(P_i)$, as the energy barrier of $E$ along the path $r$.

The minimum energy associated with a Pauli $P$ is the smallest value of $\epsilon_{\text{max}}$ across all possible paths from $0$ to $P$. This is the energy barrier of $P$, denoted as $\Delta(P)$:
\begin{eqnarray}
  \Delta(P) = \min_{r\in w(I,P)}\epsilon_{\max}(r).
\end{eqnarray}

The energy barrier of the quantum code is the minimum energy barrier over the set of nontrivial logical operators.
\begin{dfn}
   Let $S$ be a stabilizer group and $L(S)$ be the set of nontrivial logical operators. The energy barrier is  
   \begin{eqnarray}
      \Delta(H) := \min_{\ell \in L(S)}\Delta(\ell).
   \end{eqnarray}
   \label{dfn:quantum_energy_barrier}
\end{dfn}
This is the smallest energy the environment has to overcome to enact a logical operation on the encoded qubit. We can similarly define the energy barrier of classical codes by only considering the path formed by Pauli-$X$s. We shall denote this energy barrier also as $\Delta(H)$, where $H$ in this case is the parity check matrix of the classical code.

{\color{blue}\emph{Quantum LDPC codes and their energy barriers}.}- A quantum LDPC code is a stabilizer code with a sparse parity-check matrix; see \cite{zunaira2015,Breuckmann2021ldpc} for recent reviews. The sparsity parameters are $w_c$ and $w_q$, which are the maximum row and column weights of the parity check matrix, respectively. These represent the maximum weight amongst all the checks and the maximum number of checks associated with a single bit. A code is LDPC if $w_c, w_q=O(1)$.

Here, we prove a property that holds true for any quantum LDPC code. It states that the energy barrier of two logical operators equivalent under stabilizers are equal, provided that at least one of their energy barriers is larger than or equal to $w_cw_q$. For quantum LDPC codes, $w_cw_q = O(1)$. Therefore, if the energy barrier is $\Omega(1)$ for any given non-trivial logical operator, it must also be the same energy barrier for any equivalent logical operator.

We first prove the following bound.
\begin{lemma}
\label{lemma:energy_barrier_stabilizer_only}
    Let $\mathcal{C}$ be a quantum code with sparsity parameters $(w_c, w_q)$. For any stabilizer $s \in S$, 
    \begin{equation}
        \Delta(s) \leq w_cw_q.
    \end{equation}
\end{lemma}
\begin{proof}
Without loss of generality, any stabilizer $s$ can be expressed as a product $s = s_1 \cdots s_m$, where $s_1, \ldots, s_m$ are the stabilizer generators. Consider a path that applies $s_i$ in sequence, from $i=1$ to $m$. For each $s_i$, we envision applying a (sub)sequence of Paulis in the support of $s_i$. This subsequence has a length of at most $w_c$, and each Pauli in the sequence affects at most $w_q$ stabilizers. Therefore, the highest energy attained within this subsequence is at most $w_cw_q$. Once $s_i$ is applied, the energy cost becomes zero. Therefore, $\Delta(s) \leq w_c w_q$.

\end{proof}

\begin{thm}
\label{thm:energy_barrier_bound_stabilizer}
Let $\mathcal{C}$ be a $(w_c, w_q)$ quantum LDPC code. For any logical operator $L$ and stabilizer $s \in S$,
\begin{equation}
    \Delta(Ls) \leq  \max(\Delta(L), w_cw_q).
\end{equation}
\end{thm}

\begin{proof}
Without loss of generality, consider a path $r = (\ell_1, \ldots, \ell_N) \in w(0, L)$. We can consider a new path $r' \in w(0, Ls)$ by appending $r$ with a sequence $\delta_r \in w(0, s)$, forming $(\ell_1, \ldots, \ell_N, \ell_1', \ldots, \ell_M')$ where $\delta_r = (\ell_1', \ldots, \ell_M')$. By assumption,
\begin{equation}
    \epsilon_{\max}(r') = \max(\epsilon_{\max}(r), \Delta(s)).
\end{equation}
 According to Lemma~\ref{lemma:energy_barrier_stabilizer_only}, $\Delta(s) \leq w_cw_q$. It implies that for any given path $r$, there exists a path $r^{\prime}$ such that $\epsilon_{\max}(r') \leq \max(\epsilon_{\max}(r), w_cw_q)$. Choose $r$ such that $\Delta(L) = \epsilon_{\max}(r)$. Since $\Delta(Ls) \leq \epsilon_{\max}(r')$, we have $\Delta(Ls) \leq \max(\Delta(L), w_cw_q)$. 
\end{proof}
We remark that using the same logic, one can deduce $\Delta(L) = \Delta(Lss) \leq \max(\Delta(Ls), w_cw_q)$. Consequently, if $\Delta(L)\geq w_cw_q$ or $\Delta(Ls)\geq w_cw_q,$ $\Delta(L) = \Delta(Ls)$. Therefore, to determine the asymptotic scaling of the energy barrier of a quantum LDPC code, it suffices to consider the energy barrier of any \emph{fixed} complete set of logical operators. Once an energy barrier is obtained for such a set, the energy of all the other logical operators is also essentially determined, thanks to Theorem~\ref{thm:energy_barrier_bound_stabilizer}.

However, it is a priori not obvious how to choose such a set. Below, we will solve this problem for a large family of quantum codes known as the hypergraph product code~\cite{2014-Tillich}.

{\color{blue}\emph{Hypergraph product code and its logical operators}.}- 
Hypergraph product codes are CSS codes formed from two classical linear codes. Without loss of generality, let $H_1$ and $H_2$ be $r_1\times n_1$ and $r_2\times n_2$ parity check matrices, respectively. The parity-check matrix of the hypergraph product code becomes~\cite{2014-Tillich}:
\begin{equation}
\begin{aligned}
H_X &=& \left(H_1 \otimes \mathbf{I}_{n_2} \ \mathbf{I}_{r_1} \otimes H_2^T\right),  \\
H_Z &=& \left(\mathbf{I}_{n_1} \otimes H_2 \ H_1^T \otimes \mathbf{I}_{r_2}\right).
\end{aligned}
\label{eq:parity_check_matrix}
\end{equation}
Because $H_XH_Z^T = 0$, these two parity check matrices define a CSS code, with a quantum parity check matrix
\begin{equation}
    H_{(H_1,H_2)} = \left(\begin{array}{cc}
        H_X & 0 \\
         0  & H_Z
    \end{array}\right).
\end{equation}
The classical codes defined by $H_1, H_2, H_1^T, \text{ and } H_2^T$ form the parent classical codes. Without loss of generality, we will assume that $H_i$ and $H_i^T$ define codes with parameters $[n_i, k_i, d_i]$ and $[r_i, k_i^T, d_i^T]$, for $i=1, 2$. Under this assumption, the quantum code is a $\left[[n_1n_2+r_1r_2, k_1k_2 + k_1^Tk_2^T, \min(d_1,d_2, d_1^T, d_2^T)]\right]$ code~\cite{2014-Tillich}.

We now introduce a particularly useful set of logical operators, which we refer to as the \emph{canonical} logical operators~\cite{2022-Quintavalle,manes2023distance}. For the $Z$-type logical operators, consider the following operator
\begin{equation}\label{eq:logical_canonical}
    \begin{pmatrix}
        \sum_{k, j} \lambda_{kj} \overline{x}_k \otimes y_j \\
        \sum_{\ell, m} \kappa_{\ell m} a_{\ell} \otimes \overline{b}_m
    \end{pmatrix},
\end{equation}
where (i) $H_1\overline{x}_i= H_2^T\overline{b}_m=0$ and (ii) $y_j \not\in \text{Im}(H_2^T)$ and $a_{\ell} \not\in \text{Im}(H_1)$ are unit vectors. We note that $j\in \{1,\ldots, k_2 \}$ and $\ell \in \{1, \ldots, k_1^T \}$ and that the set of logical operators expressible in this form is complete~\cite{2022-Quintavalle,manes2023distance}, which means all $k_1k_2 + k_1^Tk_2^T$ logical operators are provided. A canonical logical operator is \emph{elementary} if only one of the coefficients, either $\lambda_{kj}$ or $\kappa_{\ell m}$, equals one. A similar form of canonical $X$-type logical operators can also be constructed. Because our discussion below can be applied to such operators with little change, we omit the discussion about the $X$-type logical operators.

Let $\mathcal{G}_1(V_1, C_1)$ and $\mathcal{G}_2(V_2, C_2)$ be the Tanner graphs of codes defined by $H_1$ and $H_2$, respectively. Here, $V_1$ and $V_2$ represent the set of bits, and $C_1$ and $C_2$ denote the set of checks. We use $\{v_1^i: i \in \{1,2,\cdots n_1\}\}$ (resp. $\{v_2^j: j \in \{1,2,\cdots n_2\}\}$) to refer to bit vertices in $V_1$ (resp. $V_2$), and also as length $n_1$ (resp. $n_2$) unit vectors with the $i$th (resp. $j$th) entry as $1$. The hypergraph product $\mathcal{G}_1 \times \mathcal{G}_2$ is a bipartite graph with vertex set $V \cup C$, where $V = V_1 \otimes V_2 \cup C_1 \otimes C_2$ is the qubit set and $C = C_1 \otimes V_2 \cup V_1 \otimes C_2$ is the stabilizer set.

The set of qubits can be partitioned into two subsets, $V_1 \otimes V_2$ and $C_1 \times C_2$. For $H_X$, $H_1 \otimes \mathbf{I}_{n_2}$ acts on $V_1 \otimes V_2$ and $\mathbf{I}_{r_1} \otimes H_2^T$ acts on $C_1 \otimes C_2$. Moreover, the subset $V_1 \otimes V_2$ can be further partitioned into $n_2$ subsets $\{V_1 \otimes v_2^1, V_1 \otimes v_2^2, \cdots, V_1 \otimes v_2^{n_2} \}$, where $V_1\otimes v_2^k := \{v \otimes v_2^k: v\in V_1 \}$ and $V_2 = \{v_2^1,\ldots, v_2^{n_2} \}$ [Fig.~\ref{fig:hp-qsubset-partition}].

Thus a $Z$-type Pauli operator can be expressed as a bit-string $z=\left(z^{(1)}, z^{(2)}\right)^T$, where $z^{(1)}$ is supported on the qubit subset $V_1 \otimes V_2$ with vector space $\mathbb{F}_2^{n_1} \otimes \mathbb{F}_2^{n_2}$, and $z^{(2)}$ is supported on the qubit set $C_1 \otimes C_2$ with vector space $\mathbb{F}_2^{r_1} \otimes \mathbb{F}_2^{r_2}$. Because $z^{(1)}$ and $z^{(2)}$ act on a tensor product of two vector spaces, one can view them also as a matrix. (For instance, $v_i\otimes u_j$ for unit vectors $v_i$ and $u_j$ can be viewed as a matrix whose entry is $1$ on the $i$'th row and the $j$'th column and zero elsewhere.) We call this procedure as vector reshaping, and explain a few basic facts in the Appendix. After the reshaping, $z^{(1)}$ and $z^{(2)}$ become $Z^{(1)}$ and $Z^{(2)}$ respectively. Here, $Z^{(1)}$ is an $n_1 \times n_2$ matrix, and the entry $Z^{(1)}_{i,j}$ is supported on qubit $v_1^i \otimes v_2^j$. Moreover, the $j$th column $Z^{(1)}_{j}$ is supported on qubits $\{v_1^1\otimes v_2^j, v_1^2\otimes v_2^j, \cdots, v_1^{n_1}\otimes v_2^j\}$, which we refer to as $V_1 \otimes v_2^j$. Similarly, the subset $C_1 \otimes C_2$ can be partitioned into $r_1$ rows, and $i$th row is supported on qubit subset $c_1^i \otimes C_2$ (defined similarly).

An elementary canonical logical-$Z$ operator, which is in the form $\left(\bar{x}_k \otimes y_j, 0_{r_1 r_2}\right)^T$ (where $0_{r_1 r_2}$ is the $r_1 r_2$-dimensional zero vector), is supported on the subset $V_1 \otimes v_2^j$ if $y_j=v_2^j$. In the matrix form, it is supported on the $j$th column of $Z^{(1)}$.

\begin{figure}[htbp]
   \centering
   \includegraphics[width=86 mm]{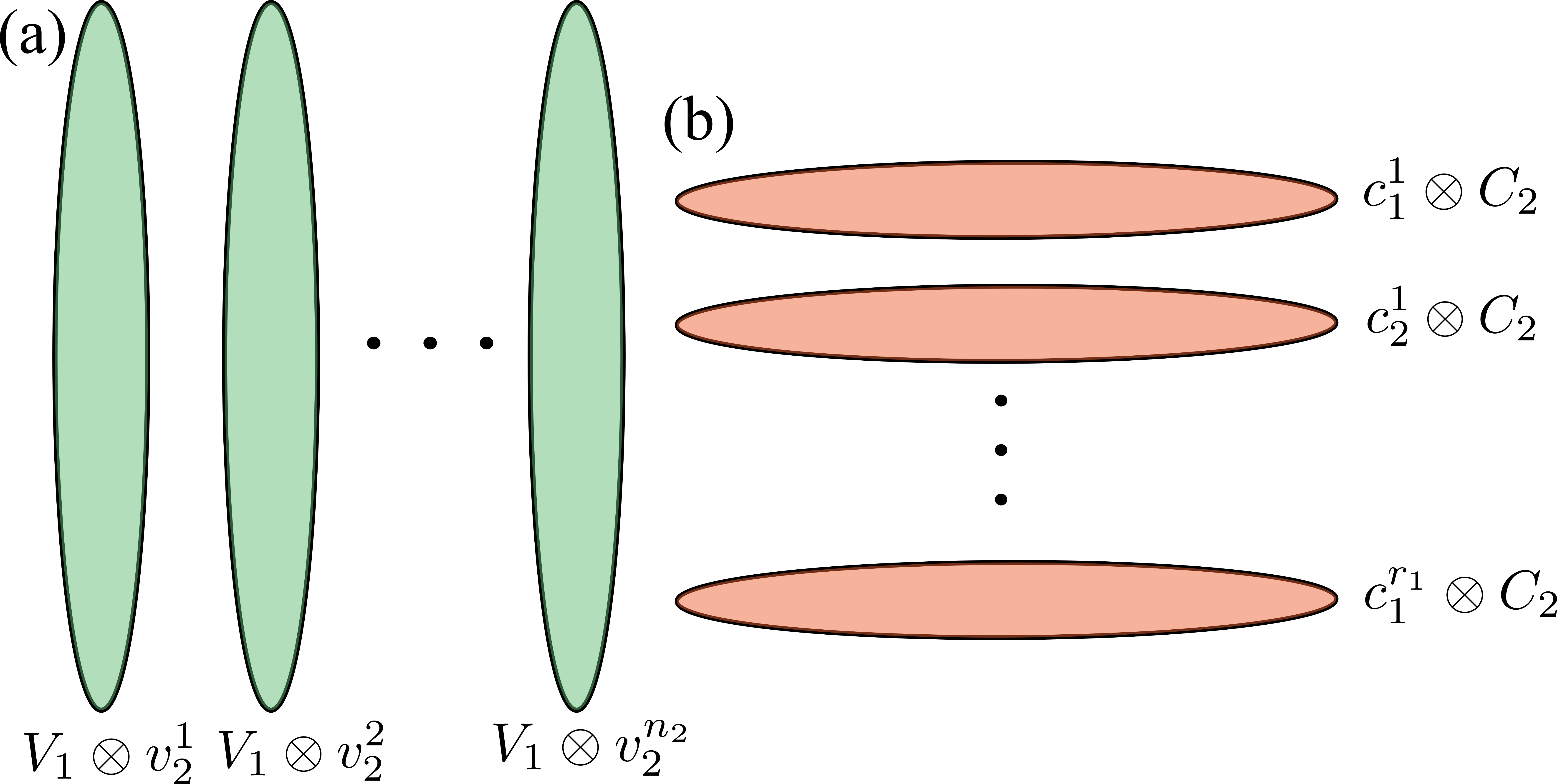}
   \caption{Graphical illustration the hypergraph product structure: (a) The qubit subset $V_1 \otimes V_2$ is partitioned into $\{V_1\otimes v_2^1, V_1\otimes v_2^2, \ldots, V_1\otimes v_2^{n_2}\}$, where $\{v_2^1, v_2^2, \ldots, v^{n_2}_2\} \in V_2$. (b) The qubit subset $C_1 \otimes C_2$ is divided into $\{c^1_1\otimes C_2, c^2_1\otimes C_2, \ldots, c^{r_1}_1\otimes C_2\}$, with $\{c_1^1, c^2_1, \cdots, c^{r_1}_1\} \in C_1$. An elementary canonical logical operator $(\bar{x}_k \otimes y_j, 0_{r_1r_2})^T$ resides on the qubit subset $V_1 \otimes v^j_2$ if $y_j = v_2^j$. Similarly, the logical operator $(0_{n_1n_2}, a_\ell \otimes \bar{b}_m)^T$ is localized on the subset $c^i_1 \otimes C_2$ if $c^i_1=a_\ell$.}
   \label{fig:hp-qsubset-partition}
\end{figure}

It would be instructive to discuss an example of the canonical logical operators. It is well-known that Kitaev's toric code~\cite{kitaev2003fault} can be viewed as two copies of 1D repetition code~\cite{2014-Tillich}. In this context, the canonical logical operators are the string operators of minimal lengths [Fig.~\ref{fig:hp-toric}].

\begin{figure}[htbp]
   \centering
   \includegraphics[width=86 mm]{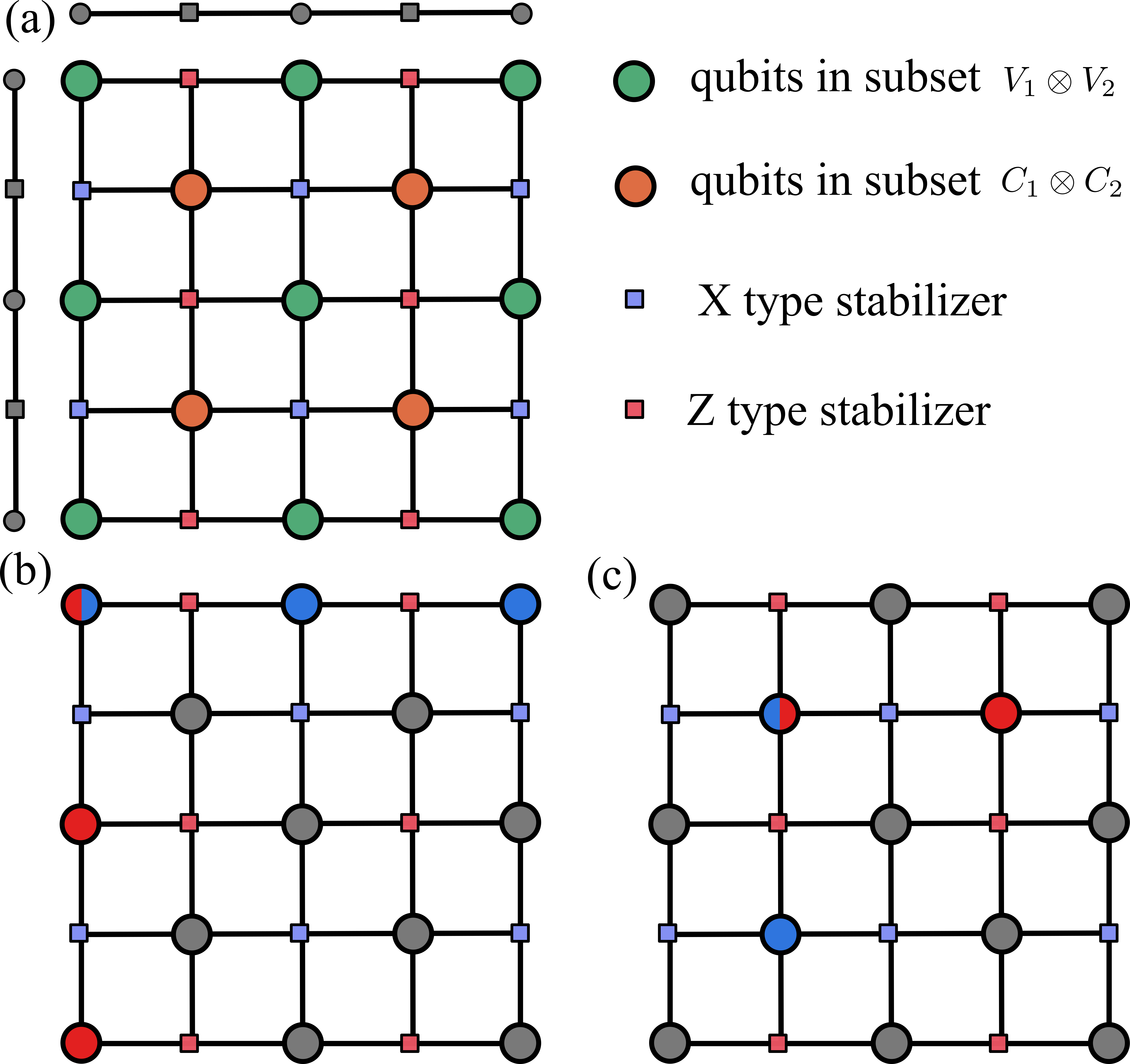}
   \caption{The hypergraph product construction of the Toric code from two repetition codes (periodic boundary conditions). (a) Delineates the two subsets of qubits, $V_1 \otimes V_2$ and $C_1 \otimes C_2$. (b) Show the logical $X$ operator (blue circles) and the logical-$Z$ operator (red circles) within the subset $V_1 \otimes V_2$. (c) Showcasing the logical $X$ (blue circles) and $Z$ (red circles) operators within the subset $C_1 \otimes C_2$.}
   \label{fig:hp-toric}
\end{figure}

{\color{blue}\emph{Energy barrier of hypergraph product code.}}- We now study relationship between the energy barrier of hypergraph product codes and their underlying classical codes. We first upper bound the energy barrier of the hypergraph product code in terms of the energy barrier of the classical codes. We then prove a matching lower bound, establishing their equivalence.

{\color{blue}\emph{(1) Upper bound}.}- The upper bound on the energy barrier can be obtained by considering a specific path from the identity to some logical operator. By choosing this logical operator to be a canonical logical operator, we obtain an upper bound. First, consider an elementary canonical logical operator of the form of $(\bar{x}_k\otimes y_j, 0_{r_1r_2})^T$, where $0_{r_1r_2}$ is the $r_1r_2$-dimensional zero vector. We will consider a path consisting of the operators of the form of $(x_i\otimes y_j)^T$ from $i=1$ to $i=N$, interpolating between $x_1=(0,\ldots, 0)^T$ to $x_N= \overline{x}_k$. The energy at the $i$th step is precisely $\mathrm{wt}(H_1 x_i)$, which is the energy of the classical code $H_1$ evaluated with respect to $x_i$. Because $x_N= \overline{x}_k$ is a codeword of $H_1$, the energy barrier along this path is at most $\Delta(H_1)$. 

Similarly, we can consider a path of the form of  $(0_{n_1n_2}, a_\ell \otimes b_i)$ from $i=1$ to $i=N$, interpolating between $b_1 = (0,\ldots, 0)^T$ to $b_N = \overline{b}_m$, where $\overline{b}_m$ is a coreword of $H_2^T$. This yields an energy barrier upper bound of $\Delta(H_2^T)$. Together, we obtain an energy upper bound of $\min(\Delta(H_1), \Delta(H_2^T))$ for the canonical logical-$Z$ operators. A similar argument can be applied to the canonical logical $X$ operators, yielding an upper bound of $\min(\Delta(H_2), \Delta(H_1^T))$.

{\color{blue}\emph{(2) Lower bound}.}- We now prove a matching lower bound for the energy barrier of the canonical logical-$Z$ operators.
\begin{proposition}\label{proposition:canonical_z_energy_barrier}
    For any nontrivial canonical logical-$Z$ operator $L$,
    \begin{equation}
        \Delta(L) \geq \min(\Delta(H_1), \Delta(H_2^T)).
    \end{equation}
\end{proposition}

We introduce two lemmas to prove this proposition. We use the following convention in the proof. A path $r=\left\{P_0, P_1, \ldots, P_F\right\}$ is said to be supported on  $U\subseteq V$ if the support of all $P_i$ in $r$ is included in $U$.

\begin{lemma}
\label{lemma:path_within_subset}
    For any elementary canonical logical-$Z$ operator $L$ supported on $V_1 \otimes v^\alpha_2$ (resp. $c_1^\beta \otimes C_2$), its energy barrier is attained by a path supported on $V_1 \otimes v^\alpha_2$ (resp. $c^\beta_2 \otimes C_2$).
\end{lemma}

\begin{lemma}
\label{lemma:energy_barrier_of_composite_logical}
    For any nontrivial canonical logical-$Z$ operator $L$, the energy barrier $\Delta(L)$ is greater than or equal to the minimum energy barrier of the elementary canonical logical-$Z$ operators.
\end{lemma}

Consider an elementary canonical logical-$Z$ operator $L$ supported on $V_1 \otimes v_2^\alpha$. Suppose $\Delta(L)$ is given by a path $r$. The main idea behind the proof of Lemma~\ref{lemma:path_within_subset} is to deform a general path $r$ into a new path $r'$, supported on  $V_1 \otimes v_2^\alpha$. Importantly, we show that the energy barrier of $r'$ is no greater than that of the original path $r$. A similar argument can be applied to prove Lemma~\ref{lemma:energy_barrier_of_composite_logical}; see the Appendix for the proofs.

Let $L =(\bar{x}_k\otimes v_2^j, 0_{r_1r_2})^T$ be an elementary canonical logical-$Z$ operator, supported on $V_1 \otimes v_2^j$. Lemma~\ref{lemma:path_within_subset} implies there is a path $r$ supported on $V_1 \otimes v_2^j$ that attains the energy barrier of $\Delta(L)$. 

Using such a path, we can prove a lower bound on the energy barrier of $L$. Without loss of generality, let $r = (P_0, P_1, \cdots, P_F)$ with $P_0 = I$ and $P_F = L$. Because this path is supported on $V_1 \otimes v_2^j$, in the binary representation, each $P_i$ becomes $u_i\otimes v_2^j$ for some $u_i \in V_1$. In particular, at the $i$'th step, the energy is $\text{wt}(H_1 u_i)$, which is exactly the energy of the classical code with respect to $u_i\in V_1$. Because $u_1=(0,\ldots, 0)$ and $u_F = \bar{x}_k$ for some codeword $\bar{x}_k$ of $H_1$, the path $(u_0, \ldots, u_F)$ must have an energy barrier of at least $\Delta(H_1)$. Similarly, for any elementary canonical logical-$Z$ operator $L$ on subset $c_1^\beta \otimes C_2$, one can show that $\Delta\left(L\right) \geq \Delta\left(H_2^T\right)$. Then Lemma~\ref{lemma:energy_barrier_of_composite_logical} imeidately implies Proposition~\ref{proposition:canonical_z_energy_barrier}.

{\color{blue}\emph{(3) Energy barrier}.}- Since the lower bound and the upper bound match, the energy barrier of the canonical logical-$Z$ operators is precisely $\min (\Delta(H_1), \Delta(H_2^T))$. Similarly, the energy barrier of the canonical logical-$X$ operators can be shown to be $\min (\Delta(H_2), \Delta(H_1^T))$. Moreover, due to Theorem~\ref{thm:energy_barrier_bound_stabilizer}, the energy barrier of any logical operator can be bounded in terms of the canonical logical operators' energy barrier and the code's sparsity parameter. Note that this conclusion applies to the hypergraph product of \emph{any} two classical codes.

If the code is LDPC and the energy barrier is $\Omega(1)$, we conclude that the energy barrier of the quantum code is exactly equal to $\min(\Delta(H_1), \Delta(H_2), \Delta(H_1^T), \Delta(H_2^T))$. 

\begin{thm}\label{thm:hypergraph_ldpc_energy_barrier}
    Let $\Delta(H)$ be the energy barrier of the hypergraph product code obtained from parity check matrices $H_1$ and $H_2$ of $O(1)$ row and column weight. If $\Delta(H_1), \Delta(H_2), \Delta(H_1^T), \Delta(H_2^T) = \Omega(1)$,
    \begin{equation}
        \Delta(H)= \min (\Delta(H_1), \Delta(H_2), \Delta(H_1^T), \Delta(H_2^T)).
    \end{equation}
\end{thm}

More precisely, for any $\left(w_c, w_q\right)$ quantum code derived from a hypergraph product, if the parent classical codes' energy barriers exceed $w_c w_q$, the quantum code's energy barrier matches the minimum energy barrier of these parent codes. This can be used to prove tight bounds on the energy barrier. For instance, a classical code whose Tanner graph is a bipartite expander graph has an extensive macroscopic energy barrier. Thus our result implies that the quantum expander code~\cite{leverrier2015quantum,fawzi2020constant} also has an extensive energy barrier. While this is already known~\cite{fawzi2020constant}, note that this argument does not rely on the expansion property of quantum expander code.

{\color{blue}\emph{Outlook}.}- We proved a tight bound on the energy barrier of the hypergraph product code, determined in terms of the energy barriers of the underlying classical codes. While it was expected that the energy barrier of the quantum code is related to its counterpart~\cite{rakovszky2024physicsgoodldpccodes}, we provided a first rigorous proof of this statement [Eq.~\eqref{eq:main_result}] to our knowledge. Looking forward, it would be interesting to study the energy barrier of the codes such as the homological product~\cite{Bravyi.2013b}, balanced product codes \cite{Breuckmann_2021} and lifted product codes \cite{panteleev2021quantum}. The generalized bicycle code \cite{kovalev2013,panteleev2021quantum}, due to its simple structure, is another natural candidate to explore. Understanding how our proof technique can be generalized to these setups and how to design efficient decoders that can leverage the energy barrier is left for future work.

\emph{Acknowledgement}- We thank Niko Breuckmann, Sergey Bravyi, Earl Campbell, Jeongwan Haah, Vedika Khemani, and Anthony Leverrier for helpful discussions. GZ acknowledges the financial support from Sydney Quantum Academy. This work was supported by the Australian Research Council Centre of Excellence for Engineered Quantum Systems (EQUS, CE170100009). IK acknowledges support from NSF Grant QCIS-FF 2013562.

\bibliographystyle{apsrev4-1-etal-title}
\bibliography{ref} 

\section{Appendix}

Throughout the appendix, we are always working in the field $\mathbb{F}_2$. Thus, all addition operations are modulo $2$ except for the computation of the weight of vectors or matrices, i.e., the function $\mathrm{wt}(\cdot)$.

\subsection{Vector reshaping}

Consider a basis $\mathcal{B}$ of the vector space $\mathbb{F}_2^{n_1} \otimes \mathbb{F}_2^{n_2}$:
\begin{eqnarray}
\mathcal{B}=\left\{a_i \otimes b_j \mid i=1, \ldots, n_1 \text { and } j=1, \ldots, n_2\right\}.
\end{eqnarray}
Then any vector $v \in \mathbb{F}_2^{n_1} \otimes \mathbb{F}_2^{n_2}$ can be written as
\begin{eqnarray}
    v=\sum_{a_i \otimes b_j \in \mathcal{B}} v_{i j}\left(a_i \otimes b_j\right)
\end{eqnarray}
for some $v_{i j} \in \mathbb{F}_2$. We call the $n_1 \times n_2$ matrix $V$ with entries $v_{i j}$ the \emph{reshaping} of the vector $v$. By this definition, if $A, B$ are respectively $m_1 \times n_1$ and $m_2 \times n_2$ matrices, then 
\begin{eqnarray}
    (A \otimes B)v \Rightarrow A V B^T.
\end{eqnarray}

Define $\mathrm{wt}(M)$ as the number of ones in the vector (or matrix) $M$. We have $\mathrm{wt}((A \otimes B)v) = \mathrm{wt}(A V B^T)$.

\subsection{Proof of Lemma~\ref{lemma:path_within_subset}}

Recall that a $Z$-type Pauli error of the hypergraph product code can be expressed as a bit-string $z=\left(z^{(1)}, z^{(2)}\right)^T$. The corresponding energy $\epsilon(z)= \text{wt}(H_Xz)$ is
\begin{equation}\label{eq:energy_expression_hgp}
    \epsilon(z) = \text{wt}\left((H_1\otimes I_{n_2})z^{(1)} + (I_{r_1}\otimes H_2^T) z^{(2)}\right).
\end{equation}
By applying vector reshaping, the energy can be written as
\begin{equation}\label{eq:energy_expression_hgp_reshape}
    \epsilon(z) = \text{wt}\left(H_1Z^{(1)} + Z^{(2)} H_2 \right),
\end{equation}
where $Z^{(1)}$ and $Z^{(2)}$ are the matrices reshaped from $z^{(1)}$ and $z^{(2)}$, respectively. The $j$th column of $Z^{(1)}$ is supported on qubit subset $V_1 \otimes v_2^j$.

Using Eq.~\eqref{eq:energy_expression_hgp_reshape}, we aim to prove a lower bound on the energy $\epsilon(z)$ [Lemma~\ref{lemma:pauli_weight_reduction}]. To that end, we shall use the following convention. Let $L_c$ be a nontrivial codeword of $H_2$. The set of columns of $Z^{(1)}$ associated with the nonzero entries of $L_c$ will play an important role. We define the index set of such columns as $\mathcal{C}(L_c)$:
\begin{equation}
    \mathcal{C}(L_c):= \{j: (L_c)_j=1 \}.
\end{equation}

Given a codeword $L_c$ of $H_2$, one can construct a vector $z^{1, s}$ from $Z^{(1)}$ by summing all the columns in the set $\mathcal{C}(L_c)$. More formally,
\begin{equation}
    z^{1,s}_i = \sum_{ j : j \in \mathcal{C}(L_c) } Z^{(1)}_{ij},
    \label{eq:prescription_weight_reduction}
\end{equation}
where the addition is modulo $2$. For example, if $L_c=110100$, we would sum columns $1, 2$, and $4$. Using Eq.~\eqref{eq:prescription_weight_reduction}, we can deform an arbitrary path to a path consisting of Paulis only supported on subset $V_1\otimes v_2^k$ for some $k$.\footnote{The precise choice of $k$ does not matter; any $k\in \mathcal{C}(L_c)$ would suffice.} In particular, we can prove an inequality between the energy of the original Pauli and the deformed Pauli, proved in Lemma~\ref{lemma:pauli_weight_reduction}. 
\begin{lemma}
\label{lemma:pauli_weight_reduction}
    \begin{equation}
        \mathrm{wt}(H_1z^{1,s}) \leqslant \mathrm{wt}\left(H_1Z^{(1)} + Z^{(2)} H_2\right).
    \end{equation}
\end{lemma}

\begin{proof}
    We prove this by contradiction. Consider the row spaces of $H_1z^{1,s}$ and $H_1Z^{(1)} + Z^{(2)} H_2$. If $\mathrm{wt}(H_1z^{1,s}) > \mathrm{wt}\left(H_1Z^{(1)} + Z^{(2)} H_2\right)$, then there exists a row $j$ such that
    \begin{eqnarray}
    \label{eq:weight_checks_rowspace}
        \mathrm{wt}((H_1z^{1,s})_{\text{row } j}) > \mathrm{wt}\left( (H_1Z^{(1)} + Z^{(2)} H_2)_{\text{row } j}\right).
    \end{eqnarray}
    We will prove that Eq.~\eqref{eq:weight_checks_rowspace} cannot be satisfied, thereby proving the claim.

    Without loss of generality, consider the $j$'th row. Since $H_1z^{1,s}$ is a $r_1 \times 1$ vector, $\mathrm{wt}((H_1z^{1,s})_{\text{row } j})$ must be either 0 or 1. If $\mathrm{wt}((H_1z^{1,s})_{\text{row } j}) = 0$, Eq.~(\ref{eq:weight_checks_rowspace}) cannot be satisfied. Therefore, we consider the $\mathrm{wt}((H_1z^{1,s})_{\text{row } j}) = 1$ case.
    
    If $\mathrm{wt}((H_1z^{1,s})_{\text{row } j}) = 1$,  $(H_1Z^{(1)})_{\text{row } j}$ must contain an odd number of ones on the columns in the set $\mathcal{C}(L_c)$. Otherwise, we would have had $\mathrm{wt}((H_1z^{1,s})_{\text{row } j}) = 0$, which is a contradiction. On the other hand, we remark that $(Z^{(2)} H_2)_{\text{row } j}$ consists of an even number of ones on the column set $\mathcal{C}(L_c)$. To see why, note that each row of $Z^{(2)}H_2$ is a linear combination of the checks in $H_2$. Because $L_c$ is a codeword of $H_2$, $(Z^{(2)} H_2)_{\text{row}, j} L_c=0$. Therefore, in the $\mathrm{wt}((H_1z^{1,s})_{\text{row } j}) = 1$ case, the number of ones in
    $H_1Z^{(1)} + Z^{(2)} H_2$ on the $j$'th row and the columns in $\mathcal{C}(L_c)$ is odd. 

    Thus we conclude $\mathrm{wt}\left( (H_1Z^{(1)} + Z^{(2)} H_2)_{\text{row } j}\right)\geq 1$. As such, Eq.~(\ref{eq:weight_checks_rowspace}) cannot be satisfied. This completes the proof.
\end{proof}

Now we are in a position to prove Lemma~\ref{lemma:path_within_subset}. We do so by identifying a path $r'=\left\{P_0', P_1', \cdots, P_F'\right\}$ that is only supported on $V_1\otimes v_2^\alpha$ while ensuring that $\epsilon_{\max }\left(r^{\prime}\right) \leqslant \epsilon_{\max }(r)$. Lemma~\ref{lemma:pauli_weight_reduction} suggests a way to deform the path $r$ to the one supported on $V_1\otimes v^\alpha_2$. 

Without loss of generality, let $r=\left\{P_0, P_1, \cdots, P_F\right\}$ be a path that $\Delta(L) = \epsilon_{\max}(r)$, with $P_0 = I$ and $P_F = L$. We consider a Pauli $P_i$ in the path $r$. It will be convenient to work in its binary representation, written as $(p^{(1)}, p^{(2)})^T$, where $p^{(1)}$ and $p^{(2)}$ represent the Paulis supported on $V_1\otimes V_2$ and $C_1\otimes C_2$, respectively. First, we remove all the Paulis supported on $C_1\otimes C_2$ by setting $p^{(2)}$ as the zero vector. Next, we apply the following transformations to $p^{(1)}$. We reshape $p^{(1)}$ and refer to the reshaped matrix as $P^{(1)}$. Let $L_c$ be a codeword of $H_2$ such that $(L_c)_{\alpha}=1$. Consider a set of columns in $P^{(1)}$ corresponding to the index set $\mathcal{C}(L_c)$. Denoting each column as $u_k$, where $k \in \mathcal{C}(L_c)$, we update the column $u_{\alpha}$ in the following way:
\begin{equation}
    u_\alpha \to u_{\alpha} + \sum_{k\in \mathcal{C}(L_c) \setminus \{ \alpha\}} u_k.
    \label{eq:column_update}
\end{equation}
Afterwards, the other columns of $P^{(1)}$ are set to the zero vector. This yields the deformed Pauli operator $P_i'$.

By construction, the resulting $P_i^{\prime}$ is supported on $V_1 \otimes v^\alpha_2$. Note that $r'$ is a valid path because $\mathrm{wt}(P_i^{\prime} P_{i+1}^{\prime}) \leq 1$ for every $i$. Also, because $P_F' = P_F = L$, $r'$ is still a path for $L$. Moreover, because $\epsilon\left(P_i^{\prime}\right) \leqslant \epsilon\left(P_i\right)$ for all $i$ [Lemma~\ref{lemma:pauli_weight_reduction}], we have $\epsilon_{\max }\left(r^{\prime}\right) \leqslant \epsilon_{\max }(r)$. Both $r'$ and $r$ are paths for $L$, by definition  $\epsilon_{\max }\left(r^{\prime}\right) \geqslant \epsilon_{\max }(r)$, we conclude $\epsilon_{\max }\left(r^{\prime}\right) = \epsilon_{\max }(r)$.  Thus, by deforming $r$, we obtained a new path $r'$ supported on $V_1 \otimes v^\alpha_2$ that yields the energy barrier $\Delta(L)$.

This argument can be applied to prove similar lower bounds for logical operators on $c^\beta_1 \otimes C_2$. To conclude, for any elementary canonical logical operator $L$ supported on $V_1 \otimes v^\alpha_2$ (resp. $c^\beta_1 \otimes C_2$), their energy barrier can be given by a path supported on $V_1 \otimes v^\alpha_2$ (resp. $c^\beta_1 \otimes C_2$).

\subsection{Proof of Lemma~\ref{lemma:energy_barrier_of_composite_logical}}

Any nontrivial canonical logical-$Z$ operator $L$ belongs to one of the following categories:
\begin{itemize}
    \item Case 1: $L$ is supported solely on the qubit subset $V_1 \otimes V_2$.
    \item Case 2: $L$ is supported solely on the qubit subset $C_1 \otimes C_2$.
    \item Case 3: $L$ is supported on both subsets.
\end{itemize}
We will focus solely on Case 1. Case 2 can be analyzed similarly by considering subsets $C_1 \otimes C_2$, while case 3 can be treated as Case 1 or 2.

Without loss of generality, let the energy barrier of $L$ be attained by a path $r = \{P_0, P_1, \cdots, P_F\}$, with $P_0 = I$ and $P_F = L$. Similar to the approach taken in Lemma~\ref{lemma:path_within_subset}, we aim to deform the path $r$ to the one supported on $V_1\otimes v_2^k$ for some $k$, such that the energy barrier of the deformed path lower bounds that of the $r$.

The deformation works in the same way as in the proof of Lemma~\ref{lemma:path_within_subset}. We describe this procedure again for the readers' convenience.  Let $L_c$ be a nontrivial codeword of $H_2$ and $\mathcal{C}(L_c)$ be its corresponding column index set. We consider a binary representation of a Pauli $P_i$, written as $(p^{(1)}, p^{(2)})^T$. As in the proof of Lemma~\ref{lemma:path_within_subset}, we remove the Paulis supported on $C_1\otimes C_2$ by setting $p^{(2)}$ as the zero vector. Next, reshape $p^{(1)}$ into a matrix $P^{(1)}$ and update its columns in the following way. Choose $\alpha \in \mathcal{C}(L_c)$. This column is updated as Eq.~\eqref{eq:column_update}. The other columns of $P^{(1)}$ are converted to zero vectors.

Thanks to Lemma~\ref{lemma:pauli_weight_reduction}, we obtain a new path $r'=\{P_0', P_1', \cdots, P_F'\}$ supported on $V_1 \otimes v_\alpha^2$ with the property $\epsilon_{\max}\left(r'\right) \leqslant \epsilon_{\max}\left(r\right)$. Note that $P_F'$, in the binary representation, is of the form $\bar{x}\otimes v_{\alpha}^2$, where $\bar{x}$ is a codeword of $H_1$. Therefore, $P_F'$ is either a nontrivial elementary canonical logical operator or the identity. In the latter case, $P_F'$ is the trivial codeword (zero vector) in the binary representation. Henceforth, we denote this as $L' = P_F'$.

If $L'$ is nontrivial, we can use the relation between the energy barriers of $L$ and $L'$:
\begin{eqnarray}
    \Delta(L) = \epsilon_{\max}(r) \geqslant \epsilon_{\max}(r') \geqslant \Delta(L').
\end{eqnarray}
Because $L'$ is an elementary logical operator, $\Delta(L)$ is greater or equal to the minimum energy barrier of elementary canonical logical operators. Thus, if $L'$ is nontrivial, the proof follows immediately.

If $L'$ is an identity, the above argument does not work. Fortunately, it turns out that for any $L'$, one can choose $L_c$ (the codeword of $H_2$ used in the current proof) such that $L'$ is not an identity.

Without loss of generality, consider a canonical logical-$Z$ operator $L$, expressed as 
\begin{equation}
    L = \begin{pmatrix} \sum_{k,j} \lambda_{kj}\bar{x}_k\otimes y_j \\ 0_{r_1r_2}
    \end{pmatrix},
\end{equation}
where (i) $H_1 \bar{x}_i =0$ and (ii) $y_j \notin \operatorname{Im}\left(H_2^T\right)$ are unit vectors. If a given path ends with $L$, its deformation (using Eq.~\eqref{eq:column_update}) yields the following operator $L'$:
\begin{equation}
    L' = \begin{pmatrix} \sum_{k} c_k\bar{x}_k\otimes y_{\alpha} \\ 0_{r_1r_2}
    \end{pmatrix},
\end{equation}
where $\alpha\in \mathcal{C}(L_c)$ and  $c_k$ is defined as 
\begin{equation}
    c_k := \sum_{j\in \mathcal{C}(L_c)} \lambda_{kj}.
\end{equation}
Note that $L'$ is trivial if and only if $c_k=0$ for all $k$. Therefore, we aim to prove that there exists a choice of $L_c$ such that $c_k=1$ for at least one $k$.

Let us prove the contrapositive. Suppose $c_k=0$ for all $k$, for any choice of $L_c$. Consider the following vector:
\begin{equation}
    u_k := \sum_{j} \lambda_{kj} y_j.
\end{equation}
Note that $c_k = u_k^T L_c$. By our assumption $c_k=0$ for any choice of $L_c$ and so the inner product of $u_k$ with any codeword of $H_2$ must be zero. On the other hand, $u_k$, if it is nonzero, must lie outside of $\text{Im}(H_2^T)$ by the definition of the $y_j$'s. Thus, $u_k$ is not an element of the row space of $H_2$. However, this is a contradiction for the following reason. For a linear code, let $H$ and $G$ be the parity check matrix and the generator matrix. Then $v^TG=0$ if and only if $v$ is a vector in the row space of $H$. In our setup, if $c_k = 0$ for any $L_c$, then $u_k^T G_2=0$.  This implies that $u_k$ must be in the row space of $H_2$, which contradicts the fact that it lies outside of $\text{Im}(H_2^T)$. To conclude, there must be at least one $k$ such that $c_k=1$. Thus, there is always a choice of $L_c$ such that $L'$ is not an identity, thereby proving the claim.

Case 2 can be analyzed similarly to Case 1 by considering rows in the subset $C_1 \otimes C_2$. For Case 3, one can treat it just as Case 1 or Case 2. For example, when treating it as Case 1, the logical operator $L$ has a nontrivial part in the subset $V_1 \otimes V_2$. One can prove there exists a codeword $L_c$ of $H_2$, such that after the deformation, the resulting $L'$ is a nontrivial elementary canonical logical operator.

\end{document}